\documentclass[journal,onecolumn]{IEEEtran}

\usepackage{amssymb}
\usepackage{caption2}
\usepackage{amsmath}
\usepackage{multirow}
\usepackage{amsthm}
\usepackage{graphicx}
\usepackage{CJK}
\usepackage{cite}
\usepackage{enumerate}
\usepackage{makecell}
\usepackage[section]{placeins}
\usepackage{bm}

\usepackage[numbers,sort&compress]{natbib}

\usepackage{tikz}
\usepackage{tikz}
\usetikzlibrary{positioning,chains,fit,shapes,calc}

\newtheorem{theorem}{Theorem}[section]
\newtheorem{lemma}[theorem]{Lemma}

\newtheorem{corollary}[theorem]{Corollary}

\newtheorem{definition}[theorem]{Definition}
\newtheorem{construction}[theorem]{Construction}
\newtheorem{remark}[theorem]{Remark}
\newtheorem{claim}[theorem]{Claim}

\newcommand{\Gaussbinom}{\genfrac{[}{]}{0pt}{}}

\begin{document}

\title{Covering Grassmannian Codes: Bounds and Constructions}

\author{Bingchen Qian, Xin Wang, Chengfei Xie and Gennian Ge
\thanks{The research of G. Ge was supported by National Key Research and Development Program of China under Grant Nos. 2020YFA0712100 and 2018YFA0704703, the National Natural Science Foundation of China under Grant No. 11971325, and Beijing Scholars Program.}
\thanks{B. Qian is with the School of Mathematical Sciences, Capital Normal University, Beijing 100048, China. He is also  with the School of Mathematical Sciences, Zhejiang University, Hangzhou 310027, Zhejiang, China (email: qianbingchen@zju.edu.cn).}
\thanks{X. Wang is with the Department of Mathematics, Soochow University, Suzhou 215005, Jiangsu, China (email: xinw@suda.edu.cn).}
\thanks{C. Xie is with the School of Mathematical Sciences, Capital Normal University, Beijing 100048, China (email: qydgl@qq.com).}
\thanks{G. Ge is with the School of Mathematical Sciences, Capital Normal University,
Beijing 100048, China (e-mail: gnge@zju.edu.cn).}
}
\maketitle

\smallskip
\begin{abstract}
Grassmannian $\mathcal{G}_q(n,k)$ is the set of all $k$-dimensional subspaces of the vector space $\mathbb{F}_q^n.$ Recently, Etzion and Zhang introduced a new notion called covering Grassmannian code which can be used in network coding solutions for generalized combination networks. An $\alpha$-$(n,k,\delta)_q^c$ covering Grassmannian code $\mathcal{C}$ is a subset of $\mathcal{G}_q(n,k)$ such that every set of $\alpha$ codewords of $\mathcal{C}$ spans a subspace of dimension at least $\delta +k$ in $\mathbb{F}_q^n.$ In this paper, we derive new upper and lower bounds on the size of covering Grassmannian codes. These bounds improve and extend the parameter range of known bounds.


\end{abstract}
\section{Introduction}
Network coding has attracted increasing attention in the last twenty years because of the seminal work of Ahlswede et al. \cite{MR1768542} and Li, Yeung and Cai \cite{LYC}, which showed that the throughput can be increased significantly by not just forwarding packets but also performing linear combinations of them.
A multicast network is a directed acyclic graph with one source node which has $h$ messages and $N$ receivers each of them demands all the $h$ messages of the source node to be transmitted in one round of a network use. For more details on network coding for multicast networks, we refer the survey \cite{MR3440232}.

The random network coding proposed in \cite{MR2300827} was an important step in the development of network coding.
 K\"{o}tter and Kschischang \cite{MR2451015} introduced an algebraic framework for error-correction in random network coding. Their model established a new notion of error-correcting codes, so-called constant-dimension codes. These are sets of $k$-dimensional subspaces of a vector space over a finite field, $k$-subspaces for short. For their purpose, they defined the subspace distance as $d_s(U,V)=\dim(U+V)-\dim(U)-\dim(V)=\dim(U)+\dim(V)-2\dim(U\cap V).$

 In \cite{MR3964845}, Etzion and Zhang generalized the subspace distance and introduced the covering Grassmannian codes for generalized combination networks. They also proved that optimal covering Grassmannian codes have minimal requirements for network coding solutions of some generalized combination networks. This family of networks was used in \cite{MR3782267}  and \cite{MR4306350} to show that vector network coding outperforms scalar linear network coding on multicast networks. An $\alpha$-$(n,k,\delta)_q^c$ covering Grassmannian code $\mathcal{C}$ ($\alpha$-$(n,k,\delta)_q^c$ code for short) consists of a set of $k$-subspaces of $\mathbb{F}_q^n$ such that every set of $\alpha$ codewords spans a subspace of dimension at least $k+\delta.$ When $\alpha=2$, the covering Grassmannian codes are equivalent to the constant-dimension codes. Denote $B_q(n, k, \delta; \alpha)$ as the maximum size of an $\alpha$-$(n,k,\delta)_q^c$ code in which there are no repeated codewords.

 The dual subspaces and the related structures were also considered in \cite{MR3964845}. A subspace packing (multiple Grassmannian code) $t$-$(n,k,\lambda)_q$ is a collection $\mathcal{C}$ of $k$-subspaces of $\mathbb{F}_q^n$ such that each $t$-subspace of $\mathbb{F}_q^n$ is contained in at most $\lambda$ blocks. Denote $A_q(n,k,t;\lambda)$ as the maximum number of $k$-subspaces in a $t$-$(n,k,\lambda)_q$ subspace packing without repeated blocks. The special case $\lambda=1$ corresponds to constant-dimension codes.  Subspace packings and covering Grassmannian codes are equivalent as it was proved in \cite{MR3964845} that
\begin{equation*}
  B_q(n,k,\delta;\alpha)=A_q(n,n-k,n-k-\delta+1;\alpha-1).
\end{equation*}

In this paper, we give some new upper and lower bounds for covering Grassmannian codes. Our main contributions are the following:

\begin{itemize}
  \item a new general upper bound via counting argument improving the previous results when $\delta>\alpha-1$ (Theorem \ref{subspace general upper bound}) and modified upper bounds by considering more configurations (Theorems \ref{37} and \ref{39}),
  \item the connection between covering Grassmanian codes and hypergraphs (Theorems \ref{upper bound by subspace cycle free}, \ref{upper bound by berge path free}, \ref{223} and \ref{last}),
  \item the constructions for all parameters of $3$-$(n,k,\delta)_2^c$ codes when $\delta>k$ (Theorems \ref{lower bound for B2(n,k,k+1;3)}, \ref{lower bound for B2(n,k,k+2;3)}, \ref{lower bound for B2(n,k,2k-1;3)} and \ref{lower bound for B2(n,k,2k;3)}),
  \item two lower bounds based on spreads (Theorem \ref{51}) and probabilistic methods (Theorems \ref{basepm} and \ref{55}) which determine the magnitude of $B_q(n,k,\delta;\alpha)$ when $\alpha-1\mid\delta-1$ (Corollary \ref{54}).
\end{itemize}

The rest of this paper is organized as follows. In Section \ref{Pre}, we introduce some basic definitions and some known results about hypergraph, and list known upper and lower bounds for covering Grassmannian codes. Two kinds of upper bounds for covering Grassmannian codes are given in Section \ref{Upp}. In Section \ref{Con}, we give several constructions for all parameters of $3$-$(n,k,k+\gamma)_2^c$ codes and one construction for $q>2$. In Section \ref{Ext}, two kinds of lower bounds are introduced, one is based on the existence of spreads and the other is based on the probabilistic methods. We conclude in the last section.

\section{Preliminaries}\label{Pre}
Let $\mathbb{F}_q$ be the finite field with $q$ elements and $\mathbb{F}_{q^m}, m\ge1$ be the extension field with $q^m$ elements. We use $\mathbb{F}_q^{m\times n}$ to denote the set of all $m\times n$ matrices over $\mathbb{F}_q,$ and $\mathbb{F}_{q^m}^n$ to denote the set of all row vectors of length $n$ over $\mathbb{F}_{q^m}.$ The rank of a matrix $A\in\mathbb{F}_q^{m\times n}$ is denoted by $rank(A).$ Denote $I_k$ and $O_k$ (or sometimes $O$ with no confusion) as the $k\times k$ identity and zero matrices, respectively. Let $A$ be a $k\times l$ matrix over $\mathbb{F}_q$, then the matrix $[I_k~A]$ can be viewed as a generator matrix of a $k$-subspace of $\mathbb{F}_q^{k+l},$ and it is called the lifting of $A$ \cite{MR2450762}.

The Grassmannian $\mathcal{G}_q(n,k)$ is the set of all $k$-subspaces of the vector space $\mathbb{F}_q^n.$ The cardinality of $\mathcal{G}_q(n,k)$ is the well-known $q$-binomial coefficient (also known as the Gaussian coefficient):
\begin{equation*}
  |\mathcal{G}_q(n,k)|=\Gaussbinom{n}{k}_q:=\prod_{i=0}^{k-1}\frac{q^n-q^i}{q^k-q^i}.
\end{equation*}
A good approximation of the $q$-binomial coefficient can be found in \cite{MR2451015}:
\[
q^{k(n-k)}\le\Gaussbinom{n}{k}_q<\tau \cdot q^{k(n-k)},
\]
where $\tau\approx3.48.$

A $k$-spread $\mathcal{S}$ is a collection of $k$-subspaces of the vector space $\mathbb{F}_q^n$, such that the subspaces in $\mathcal{S}$ intersect only trivially, and their union is $\mathbb{F}_q^n.$ It is known that the number of subspaces in a $k$-spread is $|\mathcal{S}|=\frac{q^n-1}{q^k-1}$ and $k$-spreads exist when $t\mid n$ \cite{MR63056}.

The upper bounds on the size of covering Grassmannian codes have been investigated with different objects, such as generalized combination network \cite{MR3964845,MR4306350}, subspace packing \cite{MR4149371} and independent configuration \cite{MR4173608}. For convenience, we restate the known upper bounds via the notation of covering Grassmannian codes.

In \cite{MR3964845}, Etzion and Zhang generalised the classic upper bound from a packing argument to get the following results. Some improvements which are related to this bound can be found in \cite{MR3964845,MR4149371}.
\begin{theorem} [\cite{MR3964845}]\label{upper bound of Etzion}
  If $n,k,\delta,$ and $\alpha$ are positive integers such that $1<k<n,$ $1\le\delta\le n-k$ and $2\le\alpha\le\Gaussbinom{k+\delta-1}{k}_q+1,$ then
  \begin{equation*}
    B_q(n,k,\delta;\alpha)\le\left\lfloor(\alpha-1)\frac{\Gaussbinom{n}{\delta+k-1}_q}{\Gaussbinom{n-k}{\delta-1}_q}\right\rfloor.
  \end{equation*}
\end{theorem}

In \cite{MR4306350}, Liu et al. gave the following theorem.
\begin{theorem}[\cite{MR4306350}]\label{Liu}
  If $n,k,\delta,$ and $\alpha$ are positive integers such that $1<k<n,$ $1\le\delta\le\min\{k,n-k\},$ and $2\le\alpha\le\Gaussbinom{k+\delta-1}{k}_q+1,$ then
  \[
  B_q(n,k,\delta;\alpha)\le\Gaussbinom{n-\delta}{k}_q\left(\left(\alpha-\left\lfloor\frac{k+\delta}{k}\right\rfloor+1\right)\frac{q^k-1}{q-1}-1\right)
 +\left\lfloor\frac{k+\delta}{k}\right\rfloor -1.
  \]
\end{theorem}

In \cite{MR4173608}, Cai et al. proposed a combinatorial structure called independent configuration which is equivalent to an $\alpha$-$(n,k,(\alpha-1)k)_q^c$ code and derived an upper bound.

\begin{theorem}[\cite{MR4173608}]\label{upper bound for IC}
If $n,k$ and $\alpha$ are positive integers such that $k\mid n$, then $B_q(n,k,(\alpha-1)k;\alpha)\le\frac{q^{n-\alpha k+2k}-1}{q^k-1}+\alpha-2.$
\end{theorem}

Based on the second-order Bonferroni Inequality, Etzion et al. in \cite{MR4149371} derived an upper bound for $B_q(n,2,2;3)$.
\begin{theorem}[\cite{MR4149371}]
If $2(q+1)m>{\Gaussbinom{n-2}{1}_q}$ for a positive integer $m$ and $n\geq 3$, then
$$
B_q(n,2,2;3)\leq\left\lfloor{\Gaussbinom{n}{1}_q}\frac{m(m+1)}{2(q+1)m-{\Gaussbinom{n-2}{1}_q}}\right\rfloor.
$$
\end{theorem}

In \cite{MR4149371}, Etzion et al. also gave some constructions on covering Grassmannian codes and proposed the following lower bounds on $B_q(n,k,\delta;\alpha)$ for $\delta\leq k$.
\begin{theorem}[\cite{MR4149371}]
  Let $1\le\delta\le k, k+\delta\le n$ and $2\le\alpha\le q^k+1$ be integers.
  \begin{itemize}
    \item If $n<k+2\delta,$ then
    \begin{equation*}
      B_q(n,k,\delta;\alpha)\ge(\alpha-1)q^{\max\{k,n-k\}(\min\{k,n-k\}-\delta+1)}.
    \end{equation*}
    \item If $n\ge k+2\delta,$ then for each $t$ such that $\delta\le t\le n-k-\delta,$ we have
    \begin{itemize}
      \item If $t<k$, then
      \begin{equation*}
        B_q(n,k,\delta;\alpha)\ge(\alpha-1)q^{k(t-\delta+1)}B_q(n-t,k,\delta;\alpha).
      \end{equation*}
      \item If $t\ge k,$ then
      \begin{equation*}
        B_q(n,k,\delta;\alpha)\ge(\alpha-1)q^{t(k-\delta+1)}B_q(n-t,k,\delta;\alpha)+B_q(t+k-\delta,k,\delta;\alpha).
      \end{equation*}
    \end{itemize}
  \end{itemize}
\end{theorem}

In \cite{MR4306350}, Liu et al. improved the theorem above by removing some conditions and got the following result.
\begin{theorem}[\cite{MR4306350}]
  Let $n,k,\delta$ and $\alpha$ be positive integers such that $1\le\delta\le k$, $\delta+k\le n$ and $\alpha\ge2.$ Then
  \[
  B_q(n,k,\delta;\alpha)\ge(\alpha -1)q^{\max\{k,n-k\}\left(\min\{k,n-k\}-\delta+1\right)}.
  \]
\end{theorem}

However, when $\delta>k,$ the constructions above do not work. In this paper,  we will give some constructions to derive new lower bounds for $B_q(n,k,\delta;\alpha)$ when $\delta>k.$

Finally in this section we introduce some results about hypergraph which are used in this paper. An $r$-uniform hypergraph $\mathcal{H}$ is a family of $r$-element subsets of a finite vertex set.  We usually denote its vertex set and edge set by $V(\mathcal{H})$ and $E(\mathcal{H})$, respectively.
\begin{definition}
   For a graph $F$ with vertex set $\{v_1,\ldots, v_p\}$ and edge set $\{e_1,\ldots , e_q\},$ a hypergraph
$\mathcal{H}$ contains a Berge $F$ if there exists a set of distinct vertices $\{w_1, \ldots, w_p\} \subseteq V(\mathcal{H})$ and distinct edges $\{f_1,\ldots , f_q\} \subseteq E(\mathcal{H}),$ such that if $e_i = v_\alpha v_\beta,$ then $\{w_\alpha, w_\beta\} \subseteq f_i.$ The vertices $\{w_1, \ldots, w_p\}$ are called the base vertices of the Berge $F.$
\end{definition}
We denote the Tur\'{a}n number $ex(n,BF)$ as the maximum number of edges of an $r$-uniform hypergraph on $n$ vertices which does not contain a Berge $F$. The case when $F$ is a cycle $C_k$ or a path $P_k$ with $k$ edges has been investigated in \cite{MR2900058} and \cite{MR3530631}, respectively.

\begin{theorem}[\cite{MR2900058}\label{hypergraph odd cycle free}]
For all $r,k\geq 3$, there exists a constant $c_{r,k}>0$, depending on $r$ and $k$, such that
$$
ex(n,BC_k)\leq c_{r,k}n^{1+\frac{1}{\lfloor k/2\rfloor}}.
$$
\end{theorem}

\begin{theorem}[\cite{MR3530631}]\label{hypergraph berge path free}
  Let $r\ge k\ge3,$ then
  \[
  ex(n,BP_k)\le\frac{(k-1)n}{r+1}.
  \]
\end{theorem}

A hypergraph is linear if every two edges intersect in at most one vertex.  Given a linear $r$-graph $F$ and a positive integer $n$, the linear Tur\'{a}n number $ex_L(n, F)$ is the maximum number of edges in a linear $r$-graph $\mathcal{H}$ that does not contain $F$ as a subgraph.

\begin{definition}
A linear cycle of length $k$ is a hypergraph with edges $\{e_1,e_2,\ldots,e_k\}$ such that $\forall i\in [k-1], |e_i\cap e_{i+1}|=1, |e_k\cap e_1|=1$ and $e_i\cap e_j=\emptyset$ for all other pairs $\{i,j\}, i\ne j$. We denote an $r$-uniform linear cycle of length $k$ by $C_k^r$.
\end{definition}

Extending the classic results of Bondy and Simonovits \cite{BS}, Collier-Cartaino, Graber and Jiang \cite{CGJ} gave the following results.

\begin{theorem}[\cite{CGJ}]\label{linear}
For all $r,k\geq 3$, there exists a constant $d_{r,k}>0$, depending on $r$ and $k$, such that
$$
ex_L(n,C_k^r)\leq d_{r,k}n^{1+\frac{1}{\lfloor k/2\rfloor}}.
$$
\end{theorem}

Let $f_r(n,v,e)$ denote the maximum number of edges in an $r$-uniform hypergraph on $n$
vertices, which does not contain $e$ edges spanned by $v$ vertices. When $e=r=3$ and $v=6$, the problem of determining $f_3(n,6,3)$ became later known as the $(6,3)$-problem. In 1976, Rusza and Szemer\'{e}di \cite{RS} resolved this problem by the celebrated Regularity Lemma \cite{S}. In 1986, Erd\H{o}s, Frankl and R\"{o}dl \cite{EFR} extended the result of \cite{RS} to arbitrary fixed $r$.
\begin{theorem}[\cite{EFR}]\label{63}
Let $r\geq 3$, then $n^{2-o(1)}<f_r(n,3r-3,3)=o(n^2)$.
\end{theorem}

An independent set of a hypergraph is a set of vertices containing no edges. The independence number of a hypergraph is the size of its largest
independent set. In \cite{MR1370956}, Duke, Lefmann and R\"{o}dl proved the following results about independence number.

\begin{lemma}[\cite{MR1370956}]\label{hypergraph-independent-number}
For all fixed $r\ge3$ there exists a constant $c>0$ depending only on $r$ such that every linear $r$-graph on $n$ vertices with average degree at most $d$ has an independent set of size at least $cn(\frac{\log d}{d})^{\frac{1}{r-1}}.$
\end{lemma}

\section{Upper bounds on the size of covering Grassmannian codes}\label{Upp}

\subsection{Upper bounds by counting}
In this section, we give some new upper bounds for $\alpha$-$(n,k,\delta)_q^c$ codes by counting argument.

\begin{theorem} \label{subspace general upper bound}
If $n,k,\delta$ and $\alpha$ are positive integers such that $\alpha\ge2$ and $\delta\le(\alpha-1)k$, then
$$B_q(n,k,\delta;\alpha)\le(\alpha-1)\frac{\Gaussbinom{n}{h}_q}
    {\Gaussbinom{k}{h}_q},$$
where $h=\lfloor k+1-\frac{\delta}{\alpha-1}\rfloor.$
\end{theorem}
\begin{proof}
  Let $\mathcal{C}$ be an $\alpha$-$(n,k,\delta)_q^c$ code. Let $h=\lfloor k+1-\frac{\delta}{\alpha-1}\rfloor$, which is a positive integer by the choices of $\alpha$ and $\delta.$ Denote $b_i$ the number of $h$-subspaces which are contained in exactly $i$ members of $\mathcal{C}$ for some nonnegative integer $i.$ We claim that $b_j=0$ for $j\ge\alpha.$ Indeed, we only need to show $b_\alpha=0.$ If not, assume $\{C_1,C_2,\dots,C_\alpha\}$ contain some common $h$-subspace, then it is easy to check that $\dim(span\{C_1,C_2,\dots,C_\alpha\})\le h+\alpha(k-h)<k+\delta,$ a contradiction to the definition of $\mathcal{C}.$ Now we have that
  \[
    \left\{
    \begin{array}{l}
            b_1+2b_2+\dots+(\alpha-1)b_{\alpha-1}=|\mathcal{C}|\Gaussbinom{k}{h}_q, \\
            b_0+b_1+\dots+b_{\alpha-1}=\Gaussbinom{n}{h}_q,
        \end{array}
\right.
  \]
  where the first equation is obtained by double counting the pairs $(H, C)$ with $C\in\mathcal{C}$ and $H$ being an $h$-subspace contained in $C,$ and the second equation is followed by counting the $h$-subspaces of $\mathbb{F}_q^n.$ Thus,
  \[
    |\mathcal{C}|=\frac{1}{\Gaussbinom{k}{h}_q}(b_1+2b_2+\dots+(\alpha-1)b_{\alpha-1})\le\frac{\alpha-1}{\Gaussbinom{k}{h}_q}(b_1+b_2+\dots+b_{\alpha-1})\le(\alpha-1)\frac{\Gaussbinom{n}{h}_q}
    {\Gaussbinom{k}{h}_q}.
  \]
\end{proof}

\begin{remark}
  When $\alpha,k,\delta$ and $q$ are fixed  and  $n$ tends to infinity, one can see that the upper bounds in Theorem \ref{upper bound of Etzion} and Theorem \ref{Liu} are both of the form $O(q^{kn})$. The new upper bounds in Theorem \ref{subspace general upper bound} is of the form $O(q^{hn})$. Thus when $\delta>\alpha-1$, our result improves the previous ones. When $1\leq \delta\leq \alpha-1$, the theoretical comparison between the three upper bounds is hard due to their different forms.
\end{remark}

Actually, we could do better than Theorem \ref{subspace general upper bound} by considering more configurations. We will illustrate our idea with a simple case.

\begin{theorem}\label{37}
$B_q(n,3,3;3)\leq (1+\frac{1}{2\Gaussbinom{3}{2}_q-1})\frac{\Gaussbinom{n}{2}_q}{\Gaussbinom{3}{2}_q}$.
\end{theorem}
\begin{proof}
Let $\mathcal{C}$ be a $3$-$(n,3,3)_q^c$ code and $h=2$. Denote $b_i$ as the number of $2$-subspaces which are contained in exactly $i$ members of $\mathcal{C}$ for some nonnegative integer $i.$ It was proved in Theorem \ref{subspace general upper bound} that $b_j=0$ for $j\geq 3$. For $i\in\{1,2\}$, let $\mathcal{C}_i$ be the family of $2$-subspaces which are contained in exactly $i$ members of $\mathcal{C}$. The key observation is that $b_1\geq 2(\Gaussbinom{3}{2}_q-1)b_2$. Indeed, for any $2$-subspace $V$ in $\mathcal{C}_2$, there are distinct $U_1,U_2\in\mathcal{C}$ such that $V\subset U_1$ and $V\subset U_2$. Note that, except for $V$, none of the $2$-subspaces in $U_1+U_2$ can be contained in another codeword in $\mathcal{C}$. Thus, except for $V$, all $2$-subspaces in $U_i$ are contained in $\mathcal{C}_1$ and none of these subspaces is obtained in the same way starting from another $2$-subspace in $\mathcal{C}_2$. Now we have that
  \[
    \left\{
    \begin{array}{l}
            b_1+2b_2=|\mathcal{C}|\Gaussbinom{3}{2}_q, \\
            b_0+b_1+b_2=\Gaussbinom{n}{2}_q,\\
            b_1\geq 2(\Gaussbinom{3}{2}_q-1)b_2.\\
        \end{array}
\right.
  \]
Let $a=2(\Gaussbinom{3}{2}_q-1)$, we have
\begin{equation*}
|\mathcal{C}|=\frac{1}{\Gaussbinom{3}{2}_q}(b_1+2b_2)=\frac{1}{\Gaussbinom{3}{2}_q}(\frac{a+2}{a+1}(b_1+b_2)+\frac{ab_2-b_1}{a+1})\\
   \leq \frac{\frac{a+2}{a+1}}{\Gaussbinom{3}{2}_q}(b_1+b_2)\leq (1+\frac{1}{2\Gaussbinom{3}{2}_q-1})\frac{\Gaussbinom{n}{2}_q}{\Gaussbinom{3}{2}_q}.
\end{equation*}
\end{proof}

\begin{remark}
When $q$ tends to infinity, the coefficient of the main term $\frac{\Gaussbinom{n}{2}_q}{\Gaussbinom{3}{2}_q}$ of this upper bound tends to $1$, which may provide some light for obtaining tight bounds in the asymptotic case.
\end{remark}

By the same idea, we could get the general upper bound for $B_q(n,k,\delta;3)$.
\begin{theorem}\label{39}
$B_q(n,k,\delta;3)\leq \left(1+\frac{1}{2\Gaussbinom{k}{h}_q-1}\right)\frac{\Gaussbinom{n}{h}_q}
    {\Gaussbinom{k}{h}_q},$ where $h=\lfloor k+1-\frac{\delta}{2}\rfloor.$
\end{theorem}


\subsection{Upper bounds via the connection between covering Grassmannian codes and hypergraphs}

In this subsection, we establish some connections between covering Grassmannian codes and hypergraphs.

 Let $\mathcal{C}$ be an $\alpha$-$(n,k,\delta)^c_q$ code. We consider $\mathcal{C}$ as a $\Gaussbinom{k}{1}_q$-uniform hypergraph $\mathcal{H_C}$. The vertex set $V(\mathcal{H_C})$ is the set of all $1$-subspaces of $\mathbb{F}_q^n$. A $k$-subspace in $\mathcal{C}$ corresponds to a hyperedge of size $\Gaussbinom{k}{1}_q$ whose vertices are the $1$-subspaces contained in this $k$-subspace. Thus we have $|V(\mathcal{H_C})|=\Gaussbinom{n}{1}_q$ and $|E(\mathcal{H_C})|=|\mathcal{C}|$.

\begin{theorem}\label{upper bound by subspace cycle free}
If $n,k,\delta$ and $\alpha$ are positive integers such that $\alpha\geq 3$ and $\delta\ge(\alpha-1)(k-1)$, then
  \[
    B_q(n,k,\delta;\alpha)=O\left(q^{\left(1+\frac{1}{\lfloor\alpha/2\rfloor}\right)n}\right).
  \]
\end{theorem}
\begin{proof}
   Let $\mathcal{C}$ be an $\alpha$-$(n,k,\delta)^c_q$ code. If the hypergraph $\mathcal{H_C}$ contains a Berge cycle of length $\alpha$, then denote the corresponding $k$-subspaces as $\{A_1,A_2,\dots,A_\alpha\}$. By the definition of the Berge cycle, we have $dim(span\{A_1,A_2,\ldots, A_\alpha\})\le \alpha k-\alpha<k+\delta$, which is a contradiction. So $\mathcal{H_C}$ is $BC_\alpha$-free and by Theorem \ref{hypergraph odd cycle free}, we have that
  \[
    |\mathcal{C}|=|E(\mathcal{H_C})|\le ex(|V(\mathcal{H_C})|, BC_{\alpha})= O\left(|V(\mathcal{H_C})|^{1+\frac{1}{\lfloor\alpha/2\rfloor}}\right)=O\left(q^{\left(1+\frac{1}{\lfloor\alpha/2\rfloor}\right)n}\right).
  \]
\end{proof}

\begin{theorem}\label{upper bound by berge path free}
If $n,k,\delta$ and $\alpha$ are positive integers such that $3\leq\alpha\leq \Gaussbinom{k}{1}_q$ and $\delta\ge(\alpha-1)(k-1)+1$, then
  \[
    B_q(n,k,\delta;\alpha)\leq(\alpha-1)\frac{\Gaussbinom{n}{1}_q}{\Gaussbinom{k}{1}_q+1}.
  \]
\end{theorem}
\begin{proof}
Let $\mathcal{C}$ be an $\alpha$-$(n,k,\delta)^c_q$ code. If the hypergraph $\mathcal{H_C}$ contains a Berge path of length $\alpha$, then denote the corresponding $k$-subspaces as $\{A_1,A_2,\dots,A_\alpha\}$. By the definition of the Berge path, we have $dim(span\{A_1,A_2,\ldots, A_\alpha\})\le \alpha k-\alpha+1<k+\delta$, which is a contradiction. So $\mathcal{H_C}$ is $BP_\alpha$-free and by Theorem \ref{hypergraph berge path free}, we have that
  \[
    |\mathcal{C}|=|E(\mathcal{H_C})|\le ex(|V(\mathcal{H_C})|, BP_{\alpha})\le (\alpha-1)\frac{|V(\mathcal{H_C})|}{\Gaussbinom{k}{1}_q+1}=(\alpha-1)\frac{\Gaussbinom{n}{1}_q}{\Gaussbinom{k}{1}_q+1}.
  \]
\end{proof}

Next we study a special case $B_q(n,2,2;3)$ which has been considered in \cite{MR4149371} by Bonferroni Inequality. We connect it with the $(6,3)$-problem.

\begin{theorem}\label{223}
$B_q(n,2,2;3)=o(q^{2n})$.
\end{theorem}

\begin{proof}
Let $\mathcal{C}$ be a $3$-$(n,2,2)^c_q$ code. For any $3$ distinct codewords $A_1,A_2,A_3$, we denote the corresponding edges as $\{e_1,e_2,e_3\}$. By the construction of $\mathcal{H_C}$, we have that $|e_i\cap e_j|\leq 1$ for $i\ne j$. Since $\mathcal{C}$ is a $3$-$(n,2,2)^c_q$ code, we have $dim(span\{A_1,A_2,A_3\})\geq 4$. Thus $|e_1\cup e_2\cup e_3|\geq 3{\Gaussbinom{2}{1}_q}-2$ and the hypergraph $\mathcal{H_C}$  does not contain $3$ edges spanned by $3{\Gaussbinom{2}{1}_q}-3$ vertices. By Theorem \ref{63}, we have that
  \[
    |\mathcal{C}|=|E(\mathcal{H_C})|\le f_{\Gaussbinom{2}{1}_q}(|V(\mathcal{H_C})|,3{\Gaussbinom{2}{1}_q}-3,3)=o(|V(\mathcal{H_C})|^2)=o(q^{2n}).
  \]
\end{proof}

Now, we connect the covering Grassmannian code with the linear hypergraph. We consider $\mathcal{C}$ as a $\Gaussbinom{k}{1}_q$-uniform hypergraph $\mathcal{H_C}$. The vertex set $V(\mathcal{H_C})$ is the set of all $(k-1)$-subspaces of $\mathbb{F}_q^n$. A $k$-subspace in $\mathcal{C}$ corresponds to a hyperedge of size $\Gaussbinom{k}{k-1}_q$ whose vertices are the $(k-1)$-subspaces contained in this $k$-subspace. Thus we have $|V(\mathcal{H_C})|=\Gaussbinom{n}{k-1}_q$ and $|E(\mathcal{H_C})|=|\mathcal{C}|$.

\begin{theorem}\label{last}
If $n,k,\delta$ and $\alpha$ are positive integers such that $\alpha\geq 3$ and $\delta>\alpha-2$, then
$$
B_q(n,k,\delta;\alpha)=O\left(q^{(1+\frac{1}{\lfloor\alpha/2\rfloor})(k-1)n}\right).
$$
\end{theorem}
\begin{proof}
Let $\mathcal{C}$ be an $\alpha$-$(n,k,\delta)^c_q$ code. If the hypergraph $\mathcal{H_C}$ contains a linear cycle of length $\alpha$, then denote the corresponding $k$-subspaces as $\{A_1,A_2,\dots,A_\alpha\}$. By the definition of the linear cycle, we have $dim(span\{A_1,A_2,\ldots, A_\alpha\})\le k+\alpha-2<k+\delta$, which is a contradiction. So $\mathcal{H_C}$ is  $C_{\alpha}^{\Gaussbinom{k}{1}_q}$-free and by Theorem \ref{linear}, we have that
  \[
    |\mathcal{C}|=|E(\mathcal{H_C})|\le ex_L(|V(\mathcal{H_C})|, C_{\alpha}^{\Gaussbinom{k}{1}_q})= O(|V(\mathcal{H_C})|^{1+\frac{1}{\lfloor\alpha/2\rfloor}})=O(q^{(1+\frac{1}{\lfloor\alpha/2\rfloor})(k-1)n}).
  \]
\end{proof}

In Table \ref{table1},  we list all known upper bounds of covering Grasmannian codes together with our results. In order to consider the asymptotic behavior of these bounds, we fix $k,\delta,\alpha,q$ and let $n$ tend to infinity. Relatively better bounds are in bold form.

\begin{table}[h]
\caption{Comparison of the upper bounds for $B_q(n,k,\delta;\alpha)$}\label{table1}
\begin{center}
\begin{tabular}{|c|c|c|c|c|c| p{5cm}|}
\hline
  {Restrictions on parameters} &  \cite{MR3964845}& \cite{MR4149371} & \cite{MR4306350} & \cite{MR4173608} &Our results \\ \hline
   $\delta>\alpha-1,\alpha\ge2$ & $O(q^{kn})$ & $O(q^{kn})$ &$O(q^{kn})$&- & \bm{$O(q^{hn})$}  [Theorem \ref{subspace general upper bound}] \\ \hline
  $\delta=(\alpha-1)(k-1),\alpha\ge4$ &$O(q^{kn})$ &$O(q^{kn})$&$O(q^{kn})$&-&\bm{$O(q^{(1+\frac{1}{\lfloor\alpha/2\rfloor})n})$} [Theorem \ref{upper bound by subspace cycle free}] \\ \hline
  $\delta\ge(\alpha-1)(k-1)+1,3\leq\alpha\leq \Gaussbinom{k}{1}_q$&$O(q^{kn})$ &$O(q^{kn})$&$O(q^{kn})$&-&\bm{$O(q^n)$} [Theorem \ref{upper bound by berge path free}] \\\hline
  $k\mid n,\delta=(\alpha-1)k$ & $O(q^{kn})$&$O(q^{kn})$&$O(q^{kn})$&\bm{$O(q^{n})$}& $O(q^{n})$ [Theorem \ref{upper bound by berge path free}] \\\hline
  $\alpha\geq 2k,\delta>\alpha-2$&$O(q^{kn})$&$O(q^{kn})$& $O(q^{kn})$&-&\bm{$O(q^{(1+\frac{1}{\lfloor\alpha/2\rfloor})(k-1)n})$} [Theorem \ref{last}]\\\hline
  $k=\delta=2,\alpha=3$&$O(q^{2n})$& $O(q^{2n})$&$O(q^{2n})$&-& \bm{$o(q^{2n})$} [Theorem \ref{223}]\\

   \hline

\end{tabular}
\end{center}
\end{table}

\section{Constructive lower bounds for covering Grassmannian codes}\label{Con}

To the best of our knowledge, the known constructions do not work for $\delta>k$. In this section,  we will give some constructions to derive new lower bounds for $B_q(n,k,\delta;\alpha)$ when $\delta>k.$ To be more precise, we write $\delta=k+\gamma$ where $\gamma$ is a positive integer. When $\alpha=2,$ it was proved in \cite{MR3964845} that a $2$-$(n,k,\delta)_q^c$  code is equivalent to the Grassmannian code with minimum distance $\delta$. Hence, we consider the case of $\alpha=3.$

\subsection{$\alpha=3,q=2$}
We first introduce the basic framework of our construction. We construct a matrix code $\mathcal{M}$ which is closed under subtraction and then lift it. In order to prove that the lifted code is a $3$-$(n,k,k+\gamma)_2^c$ code, it needs to show that for any three distinct matrices $M_1,M_2,M_3\in\mathcal{M},$ we have
\begin{equation*}
    rank\left(
    \begin{array}{cc}
      I_k & M_1 \\
      I_k & M_2 \\
      I_k & M_3
    \end{array}
    \right)= k +
    rank\left(
    \begin{array}{c}
      M_2-M_1  \\
      M_3-M_1
    \end{array}
    \right)\ge 2k+\gamma.
  \end{equation*}
  Since $\mathcal{M}$ is closed under subtraction, it is enough to show that for any two distinct nonzero matrices $M_1',M_2'\in \mathcal{M}$,
\begin{equation*}
    rank\left(
    \begin{array}{c}
      M_1'  \\
      M_2'
    \end{array}
    \right)\ge k+\gamma.
  \end{equation*}

We  start with the simple case of $\gamma=1$.

\begin{construction}\label{base}
Define $E_i=\left(                 
  \begin{array}{ccc}   
    O_{k\times i} & I_k & O_{k\times(n-2k-i)}\\  
  \end{array}
\right) \text{ for }  0\le i\le n-2k$ and let $\mathcal{M}_1=\{\sum\limits_{i=0}^{n-2k}\alpha_i E_i:\alpha_i\in\mathbb{F}_2,  0\le i\le n-2k\}$. It's easy to see that $\mathcal{M}_1$ is a subspace with dimension $n-2k+1$ in $\mathbb{F}_2^{k\times (n-k)}.$ Lifting the matrices in $\mathcal{M}_1$, $2^{n-2k+1}$ different matrices of size $k\times n$ are constructed. Let $\mathbb{C}_1$ be the set of rowspaces of these matrices.
\end{construction}
\begin{theorem}\label{lower bound for B2(n,k,k+1;3)}
Let $n$ and $k$ be positive integers such that $k\geq 2$ and $n\geq 2k+1$. Then $B_2(n,k,k+1;3)\ge 2^{n-2k+1}.$
\end{theorem}
\begin{proof}
 Let $\mathbb{C}_1$ be the set of $k$-subspaces obtained in Construction \ref{base}. We claim that $\mathbb{C}_1$ is a $3$-$(n,k,k+1)^c_2$ code. Given two distinct nonzero matrices $A_1,A_2\in\mathcal{M}_1$, by the construction of $\mathcal{M}_1$, it is easy to verify that $rank\begin{pmatrix}
                                   A_1 \\
                                   A_2
                                 \end{pmatrix}\leq k$ if and only if $A_1=A_2$ or one of them is a zero matrix. It contradicts to the assumption that $A_1$ and $A_2$ are two distinct nonzero matrices. Thus we have shown that $\mathbb{C}_1$ is a $3$-$(n,k,k+1)^c_2$ code. Then the conclusion follows by noting that $|\mathbb{C}_1|=2^{n-2k+1}$.
\end{proof}

Combining Theorem \ref{subspace general upper bound} and Theorem \ref{base}, we get the following result.

\begin{corollary}
$B_2(n,2,3;3)=\Theta(2^n)$.
\end{corollary}

When $\gamma>1$, the situation becomes complicated. For convenience, we define a shift operation to a matrix $A$ as follows. Let $A$ be an $a\times b$ matrix, 
then $sh(A)$ is an $a\times(a+b-1)$ matrix such that
$sh(A)_{i,j+i-1}=A_{i,j}$ if $1\le i\le a,$ $1\le j\le b,$ and $sh(A)_{i,j}=0$ otherwise.
For example, if $A=\begin{pmatrix}
                     a_1 & a_2 \\
                     a_3 & a_4
                   \end{pmatrix}$, then $sh(A)=\begin{pmatrix}
                                                 a_1 & a_2 & 0\\
                                                  0 & a_3 & a_4
                                               \end{pmatrix}$.

We give the following three constructions for the lower bounds of $B_2(n,k,k+\gamma;3)$ based on the relationship between $\gamma$ and $k$.

When $k\ge3\lceil\frac{\gamma}{2}\rceil$, let $m=\lceil\frac{\gamma}{2}\rceil$ and $\ell=\lfloor\frac{k}{m}\rfloor$. Then $2m\geq\gamma$. Define
  \begin{equation*}
    \mathcal{A}=\left\{
    \begin{pmatrix}
        a_{1} & a_{2}  & \cdots & a_{k} & b_1  \\
        a_{1} &  a_2 & \cdots &  a_k & b_2 \\
        \vdots & \vdots & \ddots & \vdots &  \vdots\\
        a_{1} & a_2 & \cdots & a_k & b_k\\
    \end{pmatrix}:\makecell{a_i\in\mathbb{F}_2 \text{ for }1\le i\le k,\\
                 a_i=0 \text{ if }i\not\equiv1 (\bmod{m})  \text{ or }i>1+(\ell-1)m, \\
                  b_i=a_{jm+1}\text{ if }jm+1\leq i\leq\min\{(j+1)m,k\},0\leq j\leq l}
    \right\}.
    \end{equation*}

\begin{construction}
Define the $k\times(t(k+1))$ matrix code $\mathcal{M}_2'=\{(A_1| A_2|\cdots| A_t):A_i\in\mathcal{A}\}$  where $t=\lfloor\frac{n-2k+1}{k+1}\rfloor.$ Let $\mathcal{M}_2=\{sh(M):M\in\mathcal{M}_2'\}.$ It's easy to see that $\mathcal{M}_2$ is a subspace with dimension $\lfloor\frac{n-2k+1}{k+1}\rfloor \lfloor\frac{k}{\lceil\gamma/2\rceil}\rfloor$ in $\mathbb{F}_2^{k\times (t(k+1)+k-1)}.$ Lifting the matrices in $\mathcal{M}_2$, $2^{\lfloor\frac{n-2k+1}{k+1}\rfloor \lfloor\frac{k}{\lceil\gamma/2\rceil}\rfloor}$ different matrices of size $k\times(t(k+1)+2k-1)$ are constructed. Let $\mathbb{C}_2$ be the set of rowspaces of these matrices.
\end{construction}

\begin{theorem}\label{lower bound for B2(n,k,k+2;3)}
Let $n,k$ and $\gamma$ be positive integers such that $k\ge3\lceil\frac{\gamma}{2}\rceil$ and $n\geq 2k+\gamma$.  Then $B_2(n,k,k+\gamma;3)\ge 2^{\lfloor\frac{n-2k+1}{k+1}\rfloor \lfloor\frac{k}{\lceil\gamma/2\rceil}\rfloor}.$
\end{theorem}

\begin{proof}
 We need to show that for every pair of distinct nonzero matrices $M_1, M_2\in \mathcal{M}_2'$, we have $rank\begin{pmatrix}
                                                                                 sh(M_1) \\
                                                                                 sh(M_2)
                                                                               \end{pmatrix}\ge k+\gamma.$ Let $i$ and $j$ be the least integers such that $(M_1)_{1i}$ and $(M_2)_{1j}$ are nonzero, respectively. Since $M_1$ and $M_2$ are nonzero, $i$ and $j$ exist. Without loss of generality, we assume that $i<j.$ Since if $i=j,$ then we can let $M_2:=M_1+M_2$ which makes $i<j.$

    If $j\ge i+2m\geq i+\gamma,$ then the first $k$ rows and the last $\gamma$ rows of $\begin{pmatrix}
                                                                                 sh(M_1) \\
                                                                                 sh(M_2)
                                                                               \end{pmatrix}$ are linearly independent, which implies that $rank\begin{pmatrix}
                                                                                 sh(M_1) \\
                                                                                 sh(M_2)
                                                                               \end{pmatrix}\ge k+\gamma.$

     It remains to prove that $rank\begin{pmatrix}
                                                                                 sh(M_1) \\
                                                                                 sh(M_2)
                                                                               \end{pmatrix}\ge k+\gamma$
     for $i<j<i+2m$. $(M_1)_{1i}\neq0$ implies that $i=\alpha(k+1)+1+\beta m$, where $0\leq\alpha\leq t-1$ and $0\leq\beta\leq \ell-1$. Similarly, $j=\lambda(k+1)+1+\mu m$, where $0\leq\lambda\leq t-1$ and $0\leq\mu\leq \ell-1$. So $j=i+(\lambda-\alpha)(k+1)+(\mu-\beta)m$. Since $i<j<i+2m$, we have that
     \begin{itemize}
       \item $\lambda=\alpha$ and $\mu=\beta+1$,
       \item $\lambda=\alpha+1$, $\mu=0$ and $\beta=\ell-1$.
     \end{itemize}
      If $\lambda=\alpha$ and $\mu=\beta+1$, then $j=i+m$. If $\lambda=\alpha+1$, $\mu=0$ and $\beta=\ell-1$, then $j=i+k+1-(\ell-1)m\geq i+m$.

     Let $i'$ and $j'$ be the largest integers such that the $i'$-th column of $M_1$ and the $j'$-th column of $M_2$ are nonzero, respectively. Since $M_1$ and $M_2$ are nonzero, $i'$ and $j'$ exist. By construction, we can assume that $i'=\alpha'(k+1)+(\beta'+1)m+k$ and $j'=\lambda'(k+1)+(\mu'+1)m+k$, where $0\leq\alpha'\leq t-1$, $0\leq\beta'\leq \ell-1$, $0\leq\lambda'\leq t-1$, and $0\leq\mu'\leq \ell-1$. Similar to $i$ and $j,$ without loss of generality, we may assume $i'\neq j'.$  Note that $j'\geq j+k-1+m\geq 2m+k$.

     \begin{description}
       \item[(1)] If $i'>j'$ and $\beta'>0$, then the first $m$ rows, the $\{\beta'm+1,\beta'm+2,\cdots,(\beta'+1)m\}$-th rows, and the last $k$ rows of $\begin{pmatrix}
                                                                                 sh(M_1) \\
                                                                                 sh(M_2)
                                                                               \end{pmatrix}$ are linearly independent.
       \item[(2)] If $i'>j'$, $\beta'=0$ and $\mu'<\ell-1$, then the first $m$ rows, the $\{(\mu'+1)m+1,(\mu'+1)m+2, \cdots, (\mu'+2)m\}$-th rows, and the last $k$ rows of $\begin{pmatrix}
                                                                                 sh(M_1) \\
                                                                                 sh(M_2)
                                                                               \end{pmatrix}$ are linearly independent.
       \item[(3)] If $i'>j'$, $\beta'=0$ and $\mu'=\ell-1$, then the first $m$ rows and the last $(k+m)$ rows of $\begin{pmatrix}
                                                                                 sh(M_1) \\
                                                                                 sh(M_2)
                                                                               \end{pmatrix}$ are linearly independent.
       \item[(4)] If $j'>i'$ and $\mu'<\ell-1,$ then the first $k$ rows, the $\{k+\mu'm+1,k+\mu'm+2, \cdots, k+(\mu'+1)m\}$-th rows, and the last $m$ rows of $\begin{pmatrix}
                                                                                 sh(M_1) \\
                                                                                 sh(M_2)
                                                                               \end{pmatrix}$ are linearly independent.
       \item[(5)] If $j'>i'$, $\mu'=\ell-1$, $\beta'>0$ and $j=i+m$, then the first $k$ rows, the $\{k+(\beta'-1)m+1,k+(\beta'-1)m+2, \cdots,k+\beta'm\}$-th rows, the $\{k+\mu'm+1, k+\mu'm+2, \cdots,k+(\mu'+1)m\}$-th rows of $\begin{pmatrix}
                                                                                 sh(M_1) \\
                                                                                 sh(M_2)
                                                                               \end{pmatrix}$ are linearly independent.
       \item[(6)]  If $j'>i'$, $\mu'=\ell-1$, $\beta'>0$ and $j=i+k+1-(\ell-1)m$, then the first $k$ rows, the $\{k+(\beta'-1)m+1, k+(\beta'-1)m+2, \cdots, k+(\beta'-1)m+2m+i-j\}$-th rows, and the last $(j-i)$ rows of $\begin{pmatrix}
                                                                                 sh(M_1) \\
                                                                                 sh(M_2)
                                                                               \end{pmatrix}$ are linearly independent.
       \item[(7)] If $j'>i'$, $\mu'=\ell-1$ and $\beta'=0$, then the first $k$ rows and the last $2m$ rows of $\begin{pmatrix}
                                                                                 sh(M_1) \\
                                                                                 sh(M_2)
                                                                               \end{pmatrix}$ are linearly independent.
     \end{description}
      Hence, when $i<j<i+2m$, we have that $rank\begin{pmatrix}
                                                                                 sh(M_1) \\
                                                                                 sh(M_2)
                                                                               \end{pmatrix}\ge k+2m\geq k+\gamma.$
\end{proof}

When $\gamma+1\leq k<3\lceil\frac{\gamma}{2}\rceil$, define a $k\times(k+1)$ matrix $C$ where $(C)_{i1}=1$ for $1\leq i\leq k$, $(C)_{i,k+1}=1$ for $1\leq i\leq k-2$, and $(C)_{ij}=0$ otherwise. Similarly, define a $k\times(k+1)$ matrix $D$ where $(D)_{i2}=1$ for $1\leq i\leq k$, $(D)_{i,k+1}=1$ for $2\leq i\leq k-1$, and $(D)_{ij}=0$ otherwise. Let $\mathcal{B}=\left\{\alpha C+\beta D:\alpha, \beta\in\mathbb{F}_2\right\}.$
\begin{construction}
Define the $k\times(t(k+1))$ matrix code $\mathcal{M}_3'=\{(B_1| B_2|\cdots| B_t):B_i\in\mathcal{B}\}$  where $t=\lfloor\frac{n-2k+1}{k+1}\rfloor.$ Let $\mathcal{M}_3=\{sh(M):M\in\mathcal{M}_3'\}.$ It's easy to see that $\mathcal{M}_3$ is a subspace with dimension $2\lfloor\frac{n-2k+1}{k+1}\rfloor$ in $\mathbb{F}_2^{k\times (t(k+1)+k-1)}.$ Lifting the matrices in $\mathcal{M}_3$, $2^{2\lfloor\frac{n-2k+1}{k+1}\rfloor}$ different matrices of size $k\times(t(k+1)+2k-1)$ are constructed. Let $\mathbb{C}_3$ be the set of rowspaces of these matrices.
\end{construction}

\begin{theorem}\label{lower bound for B2(n,k,2k-1;3)}
Let $n,k$ and $\gamma$ be positive integers such that $\gamma+1\leq k<3\lceil\frac{\gamma}{2}\rceil$ and $n\geq 2k+\gamma$. Then $B_2(n,k,k+\gamma;3)\ge 2^{2\lfloor\frac{n-2k+1}{k+1}\rfloor}.$
\end{theorem}

\begin{proof}
We need to show that for every pair of distinct nonzero matrices $M_1, M_2\in \mathcal{M}_3'$, we have $rank\begin{pmatrix}
                                                                                 sh(M_1) \\
                                                                                 sh(M_2)
                                                                               \end{pmatrix}\ge 2k-1\geq k+\gamma.$ Let $i$ and $j$ be the least integers such that $(M_1)_{1i}$ and $(M_2)_{1j}$ are nonzero, respectively. Since $M_1$ and $M_2$ are nonzero, $i$ and $j$ exist. Without loss of generality, we assume that $i<j.$ Since if $i=j,$ then we can let $M_2:=M_1+M_2$ which makes $i<j.$

    If $j\ge i+k,$ then $rank\begin{pmatrix}
                                                                                 sh(M_1) \\
                                                                                 sh(M_2)
                                                                               \end{pmatrix}=2k\ge 2k-1\geq k+\gamma.$

     It remains to prove the case of $i<j<i+k$, which occurs only when $j=i+1$ by the construction. Without loss of generality, we assume that $M_1=(O|\cdots|O|C|A_\alpha|\cdots|A_t)$ and $M_2=(O|\cdots|O|D|B_\alpha|\cdots|B_t)$, where $A_\alpha, \ldots, A_t, B_\alpha, \ldots, B_t\in\mathcal{B}$. Then the first $k$ rows and the last $k-1$ rows of $\begin{pmatrix}
                                                                                 sh(M_1) \\
                                                                                 sh(M_2)
                                                                               \end{pmatrix}$ are linearly independent, which implies that $rank\begin{pmatrix}
                                                                                 sh(M_1) \\
                                                                                 sh(M_2)
                                                                               \end{pmatrix}\geq 2k-1\geq k+\gamma.$
\end{proof}

\begin{construction}
Let $\mathcal{C}=\{O_{k}, I_{k}\}$ and define the $k\times(tk)$ matrix code $\mathcal{M}_4=\{(C_1| C_2|\cdots| C_t):C_i\in\mathcal{C}\}$, where $t=\lfloor\frac{n-k}{k}\rfloor.$ It's easy to see that $\mathcal{M}_4$ is a subspace with dimension $\lfloor\frac{n-k}{k}\rfloor$ in $\mathbb{F}_2^{k\times (tk)}.$ Lifting the matrices in $\mathcal{M}_4$, $2^{\lfloor\frac{n-k}{k}\rfloor}$ different matrices of size $k\times((t+1)k)$ are constructed. Let $\mathbb{C}_4$ be the set of rowspaces of these matrices.
\end{construction}

By the construction of $\mathcal{M}_4$, for every pair of distinct nonzero matrices $M_1, M_2\in \mathcal{M}_4$, we have $rank\begin{pmatrix}
                                                                                 M_1 \\
                                                                                 M_2
                                                                               \end{pmatrix}=2k.$ Thus, we get the following result.

\begin{theorem}\label{lower bound for B2(n,k,2k;3)}
Let $n$ and $k$ be positive integers such that $n\geq 3k$. Then $B_2(n,k,2k;3)\ge 2^{\lfloor\frac{n-k}{k}\rfloor}.$
\end{theorem}

Now we have completed the constructions for all parameters of $3$-$(n,k,k+\gamma)_2^c$ codes.

\subsection{$\alpha=3,q>2$}

When $q>2,$ since the matrix code we construct above is no longer closed under subtraction, we will give a construction using nonlinear indicators.

For a positive integer $\gamma$, let $\ell=\lfloor\frac{k}{\gamma}\rfloor$. Define
  \begin{equation*}
    \mathcal{D}=\left\{
    \begin{pmatrix}
        a_{1} & a_{2}  & \cdots & a_{k} & b_1  \\
        a_{1} &  a_2 & \cdots &  a_k & b_2 \\
        \vdots & \vdots & \ddots & \vdots &  \vdots\\
        a_{1} & a_2 & \cdots & a_k & b_k\\
    \end{pmatrix}:\makecell{a_i\in\mathbb{F}_q \text{ for each }1\le i\le k, \\
        a_i=0 \text{ if }i\not\equiv1(\bmod \gamma)\text{ or }i>1+(\ell-1)\gamma, \\
        b_i=a_{j\gamma+1}^2\text{ if }j\gamma+1\leq i\leq\min\{(j+1)\gamma,k\},0\leq j\leq \ell}
    \right\}.
    \end{equation*}
\begin{construction}
Define the $k\times(t(k+1))$ matrix code $\mathcal{M'}=\{(D_1| D_2|\cdots| D_t):D_i\in\mathcal{D}\}$  where $t=\lfloor\frac{n-2k+1}{k+1}\rfloor.$ Let $\mathcal{M}_5=\{sh(M):M\in\mathcal{M'}\}.$ It's easy to see that the size of $\mathcal{M}_5$ is $q^{\lfloor\frac{n-2k+1}{k+1}\rfloor\lfloor\frac{k}{\gamma}\rfloor}$. Lifting the matrices in $\mathcal{M}_5$, $q^{\lfloor\frac{n-2k+1}{k+1}\rfloor\lfloor\frac{k}{\gamma}\rfloor}$ different matrices of size $k\times(t(k+1)+2k-1)$ are constructed. Let $\mathbb{C}_5$ be the set of rowspaces of these matrices.
\end{construction}

\begin{theorem}\label{lower bound for Bq(n,k,k+gamma;3)}
Let $n,k$ and $\gamma$ be positive integers such that $n\geq 2k+\gamma$. Then $B_q(n,k,k+\gamma;3)\ge q^{\lfloor\frac{n-2k+1}{k+1}\rfloor \lfloor\frac{k}{\gamma}\rfloor}.$
\end{theorem}

\begin{proof}
We aim to show that for every distinct matrices $M_1, M_2, M_3\in \mathcal{M'}$, we have $rank\begin{pmatrix}
                                                                                 I_k & sh(M_1) \\
                                                                                 I_k & sh(M_2) \\
                                                                                 I_k & sh(M_3)
                                                                               \end{pmatrix}\geq2k+\gamma.$
    It is equivalent to show that for every distinct matrices $M_1, M_2, M_3\in \mathcal{M'}$, we have $rank\begin{pmatrix}
                                                                                 sh(M_2-M_1) \\
                                                                                 sh(M_3-M_1)
                                                                               \end{pmatrix}\geq k+\gamma.$

    Let $i$ and $j$ be the least integers such that $(M_2-M_1)_{1i}$ and $(M_3-M_1)_{1j}$ are nonzero, respectively. Since $M_1, M_2, M_3$ are distinct, $i$ and $j$ exist. Without loss of generality, we assume that $i\leq j.$

    If $j\ge i+\gamma,$ then the first $k$ rows and the last $\gamma$ rows of $\begin{pmatrix}
                                                                                 sh(M_2-M_1) \\
                                                                                 sh(M_3-M_1)
                                                                               \end{pmatrix}$ are linearly independent, which implies that $rank\begin{pmatrix}
                                                                                 sh(M_2-M_1) \\
                                                                                 sh(M_3-M_1)
                                                                               \end{pmatrix}\geq k+\gamma.$

     It remains to prove that $rank\begin{pmatrix}
                                                                                 sh(M_2-M_1) \\
                                                                                 sh(M_3-M_1)
                                                                               \end{pmatrix}\geq k+\gamma$
     for $i\leq j<i+\gamma$. $(M_2-M_1)_{1i}\neq0$ implies that $i=\alpha(k+1)+1+\beta\gamma$, where $0\leq\alpha\leq t-1$ and $0\leq\beta\leq \ell-1$. Similarly, $j=\lambda(k+1)+1+\mu\gamma$, where $0\leq\lambda\leq t-1$ and $0\leq\mu\leq \ell-1$. So $j-i=(\lambda-\alpha)(k+1)+(\mu-\beta)\gamma$. Since $i\leq j<i+\gamma$, we have that $\alpha=\lambda$ and $\beta=\mu$. Without loss of generality, we assume that $M_2-M_1=(O|\cdots|O|B_{\alpha+1}-A_{\alpha+1}|B_{\alpha+2}-A_{\alpha+2}|\cdots|B_t-A_t)$ and $M_3-M_1=(O|\cdots|O|C_{\alpha+1}-A_{\alpha+1}|C_{\alpha+2}-A_{\alpha+2}|\cdots|C_t-A_t)$, where $A_{\alpha+1}, \ldots, A_t, B_{\alpha+1}, \ldots, B_t, C_{\alpha+1}, \ldots, C_t\in\mathcal{D}$, $0\neq(M_2-M_1)_{1i}=(B_{\alpha+1}-A_{\alpha+1})_{1, 1+\beta\gamma}$, and $0\neq(M_3-M_1)_{1i}=(C_{\alpha+1}-A_{\alpha+1})_{1, 1+\beta\gamma}$. For convenience, put $a=(A_{\alpha+1})_{1, 1+\beta\gamma}, b=(B_{\alpha+1})_{1, 1+\beta\gamma}$ and $c=(C_{\alpha+1})_{1, 1+\beta\gamma}$. Then $c-a=s(b-a)$ for some $s\in\mathbb{F}_q^*$, and consider the matrix $(M_3-M_1)-s(M_2-M_1)$. Let $j'$ be the least integer such that $((M_3-M_1)-s(M_2-M_1))_{1j'}$ is nonzero.

     If $j'$ is of the form $\lambda'(k+1)+1+\mu'\gamma$ for some $\lambda<\lambda'\leq t-1$ and $0\leq\mu'\leq \ell-1$, then $j'>j$. Thus $j'\geq i+\gamma$. In this case, the first $k$ rows and the last $\gamma$ rows of $\begin{pmatrix}
                                                                                 sh(M_2-M_1) \\
                                                                                 sh((M_3-M_1)-s(M_2-M_1))
                                                                               \end{pmatrix}$ are linearly independent.

     If $j'$ is not of the form $\lambda'(k+1)+1+\mu'\gamma$ and $s=1$, then $M_3-M_1=M_2-M_1$, which means that $M_3=M_2$, a contradiction.

     If $j'$ is not of the form $\lambda'(k+1)+1+\mu'\gamma$ and $s\neq1$, then consider $(A_{\alpha+1})_{1+\beta\gamma, k+1}, (B_{\alpha+1})_{1+\beta\gamma, k+1}$ and $(C_{\alpha+1})_{1+\beta\gamma, k+1}$. By the construction of $\mathcal{D}$, $(A_{\alpha+1})_{1+\beta\gamma, k+1}=a^2, (B_{\alpha+1})_{1+\beta\gamma, k+1}=b^2$ and $(C_{\alpha+1})_{1+\beta\gamma, k+1}=c^2$. Recall that we have $c-a=s(b-a)\neq0$. Then $c^2-a^2-s(b^2-a^2)=(c-a)(c-b)$. Since $s\neq1$, $c\neq b$, then $(c-a)(c-b)\neq0$. Thus $((M_3-M_1)-s(M_2-M_1))_{m, (\alpha+1)(k+1)}$ is nonzero for $1+\beta\gamma\leq m\leq(\beta+1)\gamma$. Therefore, the first $k$ rows, the $\{k+1+\beta\gamma, k+2+\beta\gamma, \cdots, k+(\beta+1)\gamma\}$-th rows of $\begin{pmatrix}
                                                                                 sh(M_2-M_1) \\
                                                                                 sh((M_3-M_1)-s(M_2-M_1))
                                                                               \end{pmatrix}$ are linearly independent, which implies that $rank\begin{pmatrix}
                                                                                 sh(M_2-M_1) \\
                                                                                 sh(M_3-M_1)
                                                                               \end{pmatrix}\geq k+\gamma.$
\end{proof}

\section{Existence of lower bounds for covering Grassmannian codes}\label{Ext}
In this section, we will derive two kinds of lower bounds for covering Grassmannian codes. The first one is based on the existence of the spread and the second one based on the probabilistic method.
\subsection{Greedy algorithm}
In this subsection, we give a lower bound for $B_q(n,k,(\alpha-1)k;\alpha)$ when $k\mid n$ and $n$ is large enough by a greedy algorithm.

We first describe the algorithm for $B_q(n,k,2k;3)$.
Let $k$ and $n$ be integers such that $k\mid n.$ Let $V$ be a vector space of dimension $n$ over $\mathbb{F}_q.$  Denote the $k$-spread by $\mathcal{S}$ with $|\mathcal{S}|=\frac{q^n-1}{q^k-1}.$ Choose two elements in $\mathcal{S},$ say $V_1$ and $V_2.$ So $V_1$ and $V_2$ span a $2k$-subspace of $V$ which contains $\frac{q^{2k}-1}{q-1}$ $1$-subspaces. Since each $1$-subspace is contained in at most one element of $\mathcal{S},$ we throw away the elements in $\mathcal{S}$ which contain at least one $1$-subspace of $V_1\oplus V_2$ and denote the left by $\mathcal{S}_1.$ If $|\mathcal{S}|>\frac{q^{2k}-1}{q-1}$, then $|S_1|>0$. We can pick one element from $S_1$, say $V_3,$ which together with $V_1$ and $V_2$ spans a $3k$-subspace of $V.$  Then we throw away the elements in $\mathcal{S}_1$ which contain at least one $1$-subspace of $V_1\oplus V_3$ or $V_2\oplus V_3$ and denote the left by $\mathcal{S}_2.$ If $|\mathcal{S}_1|>2\frac{q^{2k}-1}{q-1},$ then $|S_2|>0$. We can pick one element from $S_2$, say $V_4$, which together with $V_1,V_2$ and $V_3$ satisfies the property that any three of them span a $3k$-subspace. Iteratively, if $|\mathcal{S}_{m-3}|>(m-2)\frac{q^{2k}-1}{q-1},$ we can always find a $V_m$ such that $V_1,\ldots,V_m$ form a $3$-$(n,k,2k)_q^c$ code. Now we could maximize $m$ such that $\frac{q^n-1}{q^k-1}-\sum\limits_{i=1}^{m-2}i\frac{q^{2k}-1}{q-1}\ge 1,$ to give a lower bound for $B_q(n,k,2k;3).$

For $\alpha>3,$ we can use a similar strategy. We also start with a spread $\mathcal{S}$ of size $\frac{q^n-1}{q^k-1}.$ Pick two elements $V_1$ and $V_2$ of $\mathcal{S},$ throw away the elements in $\mathcal{S}$ which contain at least one $1$-subspace of $V_1\oplus V_2,$ and denote the left by $\mathcal{S}_1.$ Thus each element in $\mathcal{S}_1$ together with $V_1$ and $V_2$ spans a $3k$-subspace. Pick  one element from $\mathcal{S}_1,$ say $V_3,$ throw away all elements in $\mathcal{S}_1$ which contain at least one $1$-subspace of $V_1\oplus V_2\oplus V_3,$ and denote the left by $\mathcal{S}_2.$ Thus each element in $\mathcal{S}_2$ together with $V_1,V_2$ and $V_3$ spans a $4k$-subspace. Iteratively, if $|\mathcal{S}|>\frac{q^{(\alpha-1)k}-1}{q-1}$, we can pick out $V_1,V_2,\ldots,V_{\alpha-1}$ such that they span an $(\alpha-1)k$-subspace and get $\mathcal{S}_{\alpha-2}$ such that each element in $\mathcal{S}_{\alpha-2}$ together with $V_1,\ldots,V_{\alpha-1}$ spans an $\alpha k$-subspace. Then pick one element from $\mathcal{S}_{\alpha-2},$ say $V_{\alpha},$ and throw away all elements in $\mathcal{S}_{\alpha-2}$ which contain at least one $1$-subspace of $V_1\oplus\cdots\oplus\tilde{V_i}\oplus\cdots\oplus V_{\alpha}$ for $i=1,2,\ldots,\alpha-1,$ where $V_1\oplus\cdots\oplus\tilde{V_i}\oplus\cdots\oplus V_{\alpha}$ means $V_i$ does not occur in the expression. Denote the left by $\mathcal{S}_{\alpha-1}.$ Thus each element in $\mathcal{S}_{\alpha-1}$ together with $V_1,\ldots,V_\alpha$ forms an $\alpha$-$(n,k,(\alpha-1)k)_q^c$ code. Pick one element, say $V_{\alpha+1},$ from $\mathcal{S}_{\alpha-1}$ and throw away all elements in $\mathcal{S}_{\alpha-1}$ which contain at least one $1$-subspace of $V_1\oplus\cdots\oplus\tilde{V_i}\oplus\cdots\oplus\tilde{V_j}\oplus\cdots\oplus V_{\alpha+1},$ where $i\neq j\in\{1,2,\ldots,\alpha\},$ denote the left by $\mathcal{S}_{\alpha}.$  Thus each element in $\mathcal{S}_{\alpha}$ together with $V_1,\ldots,V_{\alpha+1}$ forms an $\alpha$-$(n,k,(\alpha-1)k)_q^c$ code. Iteratively, if $|\mathcal{S}_{m-2}|>|\mathcal{S}_{m-3}|-{m-2\choose \alpha-2}\frac{q^{(\alpha-1)k}-1}{q-1},$ then we can always pick a $V_m$ from $\mathcal{S}_{m-2}$ such that $V_1,\ldots,V_m$ form an $\alpha$-$(n,k,(\alpha-1)k)_q^c$ code. Therefore, we could maximize $m$ such that
$\frac{q^n-1}{q^k-1}-\sum\limits_{i=0}^{m-\alpha}{\alpha-2+i\choose \alpha-2}\frac{q^{(\alpha-1)k}-1}{q-1}\ge 1$, to give a lower bound for $B_q(n,k,(\alpha-1)k;\alpha).$

\begin{theorem}\label{51}
Let $n,k$ and $\alpha$ be positive integers such that $k\mid n$, $\alpha\geq 3$ and $n\geq \alpha k$. Then $B_q(n,k,(\alpha-1)k;\alpha)\geq q^{\frac{n-\alpha k+1}{\alpha-1}}$.
\end{theorem}

\begin{remark}
  When $k\nmid n,$ one can use the same strategy by using the existence of partial spreads. For simplicity, we omit the details here.
\end{remark}

\subsection{Probabilistic lower bounds}
In this subsection, we will give two probabilistic lower bounds for $B_q(n,k,\delta;\alpha).$ We first give the following easy one.
\begin{theorem}\label{basepm}
Let $k,\delta$ and $\alpha$ be positive integers and $q$ be a prime power. Then there exists a constant $c$, depending only on $q,k,\delta,\alpha$ with the following property. For any sufficiently large $n$, there exists an $\alpha$-$(n,k,\delta)_q^c$ code with size at least $cq^{(k-\frac{\delta-1}{\alpha-1})n}$.
\end{theorem}
\begin{proof}
  Set $p=Cq^{-\frac{\delta-1}{\alpha-1}n}$ for some constant $C=C(q,k,\delta,\alpha)$ and pick each $k$-subspace of an $n$-dimensional space $V$ independently with probability $p.$ Let the random variable $X$ denote the number of subspaces picked, then
  \[
  \mathbf{E}[X]=p\Gaussbinom{n}{k}_q=C_0q^{(k-\frac{\delta-1}{\alpha-1})n},
  \]
   where $C_0=C_0(q,k,\delta,\alpha)$ is independent with $n.$
   We say a bad event happens if some $\alpha$ $k$-subspaces span a subspace of dimension at most $k+\delta-1.$ Let $Y$ denote the number of bad events, we have
   \[
   \mathbf{E}[Y]\leq\sum_{i=k+1}^{k+\delta-1}p^{\alpha}\Gaussbinom{n}{i}_q{\Gaussbinom{i}{k}_q \choose \alpha}.
   \]
   By the evaluation of $p,$ we have $\mathbf{E}[X]\geq 2\mathbf{E}[Y].$ By the deletion method, whenever a bad event happens, we remove one $k$-subspace and finally get the desired result.
\end{proof}

Combining Theorem \ref{subspace general upper bound} and Theorem \ref{basepm}, we get the following result.

\begin{corollary}\label{54}
If $\alpha-1\mid\delta-1$, then $B_q(n,k,\delta;\alpha)=\Theta(q^{(k-\frac{\delta-1}{\alpha-1})n}).$
\end{corollary}

When $\alpha-1\nmid\delta-1$, we can  do a little better than the results above. The following theorem is a $q$-analogue of set system in \cite{MR4118679}, and the proof is also similar. For completeness, we include it here.

\begin{theorem}\label{55}
  If $\alpha\geq 3$ and $\gcd(\alpha-1,\delta-1)=1,$ then $
    B_q(n,k,\delta;\alpha)=\Omega\left(q^{(k-\frac{\delta-1}{\alpha-1})n}n^{\frac{1}{\alpha-1}}\right)$.
\end{theorem}
\begin{proof}
  Set $p=\Theta(q^{(-\frac{\delta-1}{\alpha-1}+\epsilon)n}).$ Denote $V=\mathbb{F}_q^n$ and let $\mathcal{C}_0$ be a set of $k$-subspaces by picking each element of $\Gaussbinom{V}{k}_q$ with probability $p.$ Let $X$ denote the number of elements of $\mathcal{C}_0.$ Then
  \begin{equation}\label{E[X]}
  \mathbf{E}[X]=p\Gaussbinom{n}{k}_q=\Theta(q^{(k-\frac{\delta-1}{\alpha-1}+\epsilon)n}).
  \end{equation}
  For $2\le i\le \alpha-1,$ let $\mathcal{Y}_i$ be the collection of all $i$ distinct elements of $\mathcal{C}_0$ which span a subspace of dimension at most $ik-f(i),$ where $f(i)$ will be determined later. Let $Y_i$ denote the size of $\mathcal{Y}_i.$ Then
  \begin{equation}\label{E[Y_i]}
    \mathbf{E}[Y_i]=O\left(p^i\Gaussbinom{n}{ik-f(i)}_q{\Gaussbinom{ik-f(i)}{k}_q\choose i}\right)=O\left(q^{(ik-f(i)-\frac{\delta-1}{\alpha-1}i+\epsilon i)n}\right),
  \end{equation}
  where the first estimation follows that there are at most $O\left(\Gaussbinom{n}{ik-f(i)}_q\right)$ subspaces of dimension at most $ik-f(i)$ and each contains at most $O\left(\Gaussbinom{ik-f(i)}{k}_q\right)$ $k$-subspaces, and for every fixed choice the probability of having exactly $i$ elements is $O(p^i).$

  We say that $\alpha$ distinct elements of $\mathcal{C}_0$ form a bad $\alpha$-system if they span a subspace of dimension at most $k+\delta-1.$ For each $2\le i\le\alpha-1,$ let $\mathcal{P}_i$ be the collection of unordered pairs of bad $\alpha$-systems which share precisely $i$ elements such that these $i$ elements span a subspace of dimension at least $ik-f(i)+1.$ For $\{P,P'\}\in\mathcal{P}_i,$ it is easy to verify that
  \[
   |P\cup P'|=2\alpha-i
  \]
  and
  \[
    \dim(P+P')\le 2(k+\delta-1)-(ik-f(i)+1).
  \]
  Let $P_i$ denote the size of $\mathcal{P}_i,$ then

  \begin{align}\label{E[P_i]}
    \mathbf{E}[P_i]=&O\left(p^{2\alpha-i}\Gaussbinom{n}{2(k+\delta-1)-(ik-f(i)+1)}_q{\Gaussbinom{2(k+\delta-1)-(ik-f(i)+1)}{k}_q\choose 2\alpha-i}\right)\nonumber\\
                    =&O\left(q^{(2(k+\delta-1)-(ik-f(i)+1)+(2\alpha-i)(-\frac{\delta-1}{\alpha-1}+\epsilon))n}\right),
  \end{align}
  where the first estimation follows the fact that  there are at most $O(\Gaussbinom{n}{2(k+\delta-1)-(ik-f(i)+1)}_q$ subspaces of dimension at most $2(k+\delta-1)-(ik-f(i)+1)$ and each contains at most $O(\Gaussbinom{2(k+\delta-1)-(ik-f(i)+1)}{k}_q)$ $k$-subspaces, and for every fixed choice the probability of having exactly $2\alpha-i$ elements is $O(p^{2\alpha-i}).$

  Finally, let $Z$ denote the number of bad $\alpha$-systems in $\mathcal{C}_0,$ similarly
  \begin{equation}\label{E[Z]}
    \mathbf{E}[Z]=O\left(p^{\alpha}\Gaussbinom{n}{k+\delta-1}_q{\Gaussbinom{k+\delta-1}{k}_q\choose \alpha}\right)=O\left(q^{(k+\delta-1+\alpha(-\frac{\delta-1}{\alpha-1}+\epsilon))n}\right).
  \end{equation}

  From (\ref{E[X]}), (\ref{E[Y_i]}) and (\ref{E[P_i]}), it is not hard to see that $\mathbf{E}[Y_i]=o(\mathbf{E}[X])$ and $\mathbf{E}[P_i]=o(\mathbf{E}[X])$ hold if and only if
    \begin{equation}\label{f(i)condition}
    \left\{
    \begin{array}{l}
            f(i)>(i-1)(k-\frac{\delta-1}{\alpha-1})+\epsilon(i-1), \\
            f(i)<(i-1)(k-\frac{\delta-1}{\alpha-1})+1-\epsilon(2\alpha-i+1).
        \end{array}
\right.
  \end{equation}

   Let $a=\min_{2\le i\le\alpha-1}\left\{\frac{1}{i-1}\left(f(i)-\left(k-\frac{\delta-1}{\alpha-1}\right)\right),\frac{1}{2\alpha-i+1}\left((i-1)\left(k-\frac{\delta-1}{\alpha-1}\right)+1-f(i)\right)\right\}$.

    Observe that there is $\epsilon\in(0,a)$ satisfying (\ref{f(i)condition}) if and only if for each $2\le i\le\alpha-1,$ there exists an integer $f(i)$ such that
    \begin{equation}\label{f(i)finalcondition}
      (i-1)(k-\frac{\delta-1}{\alpha-1})<f(i)<(i-1)(k-\frac{\delta-1}{\alpha-1})+1.
    \end{equation}

   By the integrality of $f(i),$ (\ref{f(i)finalcondition}) holds if and only if
   \begin{equation*}
     \alpha-1\nmid(i-1)(\delta-1)
   \end{equation*}
   for each $2\le i\le\alpha-1.$ It is easy to verify that those $\alpha-2$ indivisibility conditions hold simultaneously if and only if
   \begin{equation*}
     \gcd(\alpha-1,\delta-1)=1.
   \end{equation*}
  Under this condition, it suffices to choose $f(i)=\left\lceil(i-1)\left(k-\frac{\delta-1}{\alpha-1}\right)\right\rceil$ for $2\le i\le\alpha-1$
 and an arbitrary $\epsilon\in(0,a).$

 Now we get a subset $\mathcal{C}_1$ from $\mathcal{C}_0$ as follows. For every $2\le i\le\alpha-1,$ delete one element from each member of $\mathcal{Y}_i,$ and one element from $P\cup P'$ for each pair $\{P,P'\}\in\mathcal{P}_i.$ By linearity of expectation, we have the following claim about $\mathcal{C}_1$.
 \begin{claim}
   \begin{enumerate}[(i)]
     \item $\mathbf{E}[|\mathcal{C}_1|]=\Theta(q^{(k-\frac{\delta-1}{\alpha-1}+\epsilon)n})$,
     \item the expected number of bad $\alpha$-systems contained in $\mathcal{C}_1$ is at most $O(q^{(k+\delta-1+\alpha(-\frac{\delta-1}{\alpha-1}+\epsilon))n})$,
     \item for each $2\le i\le\alpha,$ any $i$ distinct members in $\mathcal{C}_1$ span a subspace of dimension at least $ik-f(i)+1$,
     \item any two bad $\alpha$-systems in $\mathcal{C}_1$ can share at most one member.
   \end{enumerate}
 \end{claim}
 \begin{proof}
   Since
   \begin{align*}
     \mathbf{E}[|\mathcal{C}_1|] &=\mathbf{E}[X]-\sum_{i=2}^{\alpha-1}\mathbf{E}[Y_i]-\sum_{i=2}^{\alpha-1}\mathbf{E}[P_i] \\
      &=  \mathbf{E}[X]-o(\mathbf{E}[X])=\Theta(\mathbf{E}[X]),
   \end{align*}
   (i) is proved; (ii) follows from (\ref{E[Z]}) and that deleting members in $\mathcal{C}_0$ does not increase the number of bad $\alpha$-systems; (iii) holds since $\mathcal{C}_1$ does not contain any member of $\mathcal{Y}_i$ for any $2\le i\le\alpha-1;$ As for (iv), assume to the contrary that $\mathcal{C}_1$ contains two bad $\alpha$-systems sharing $i$ members for some $2\le i\le\alpha-1.$ By (iii), these $i$ members span a subspace of dimension at least $ik-f(i)+1,$ so they belong to $\mathcal{P}_i,$ which is a contradiction.
 \end{proof}

 Now we construct an auxiliary $\alpha$-uniform hypergraph $\mathcal{G}\subset{\mathcal{C}_1 \choose \alpha}$ as follows:
  \begin{itemize}
    \item the vertex set of $\mathcal{G}$ is formed by the members of $\mathcal{C}_1;$
    \item $\alpha$ vertices of $\mathcal{G}$ form an edge if and only if the corresponding $\alpha$ members of $\mathcal{C}_1$ form a bad $\alpha$-system.
  \end{itemize}
  It is easy to verify that the following hold:
  \begin{itemize}
    \item $\mathcal{G}$ is linear (by (iv));
    \item $\mathcal{G}$ has at least $\Theta(q^{(k-\frac{\delta-1}{\alpha-1}+\epsilon)n})$ vertices (by (i)) and at most $O(q^{(k+\delta-1+\alpha(-\frac{\delta-1}{\alpha-1}+\epsilon))n})$ edges (by (ii));
    \item  $d(\mathcal{G}),$ the average degree of $\mathcal{G},$ is at most
    \begin{equation*}
      d(\mathcal{G})=O\left(\frac{\alpha|E(\mathcal{G})|}{|V(\mathcal{G})|}\right)=O\left(q^{(\alpha-1)\epsilon n}\right).
    \end{equation*}
  \end{itemize}
  If follows from Lemma \ref{hypergraph-independent-number} that $\mathcal{G}$ has an independent set of size at least
  \begin{equation*}
    \Omega(|V(\mathcal{G})|\left(\frac{\log d(\mathcal{G})}{d(\mathcal{G})}\right)^{\frac{1}{\alpha-1}})=\Omega\left(q^{(k-\frac{\delta-1}{\alpha-1})n}n^{\frac{1}{\alpha-1}}\right).
  \end{equation*}
  Since an independent set in $\mathcal{G}$ corresponds an $\alpha$-$(n,k,\delta)_q^c$ code, the result follows.
  \end{proof}

\section{Conclusions}\label{conclu}
Motivated by an application in generalized combination networks, the covering Grassmannian codes are considered in this paper. First, we give a general upper bound by counting argument which improves the known results when $\delta>\alpha-1$. Then based on a connection between hypergraphs and covering Grassmannian codes, we obtain several new upper bounds. Specially, we establish the connection between covering Grassmannian codes and the $(6,3)$-problem. We also give the constructions for all parameters of $3$-$(n,k,\delta)_2^c$ codes when $\delta>k$ and provide two lower bounds based on spreads and probabilistic methods which determine the magnitude of $B_q(n,k,\delta;\alpha)$ when $\alpha-1\mid\delta-1$.

We conclude our paper by summarizing all known constructive results in Table \ref{table2}. From Table \ref{table2}, it can be obtained that the constructive lower bounds are still far from satisfactory even compared with the probabilistic ones. It would be of interest to narrow the gap by new methods.

\begin{table}[h]
\caption{Constructive Lower Bounds}\label{table2}
\begin{center}
\begin{tabular}{|c|c|c|c| p{5cm}|}
\hline
                    & Constructive lower bounds &\makecell{Probabilistic\\lower bounds} &  Upper bounds [Theorem \ref{subspace general upper bound}] \\ \hline
                    \makecell{$B_q(n,k,\delta;\alpha)$\\$1\leq\delta\leq k$}&$(\alpha-1)q^{(k-\delta+1)(n-k)}$ \cite{MR4306350}&$\Omega(q^{(k-\frac{\delta-1}{\alpha-1})n})$&$O(q^{\lfloor k+1-\frac{\delta}{\alpha-1}\rfloor n})$\\\hline
                    \makecell{$B_q(n,k,\delta;\alpha)$\\$1\leq\delta\leq k$}&$\Omega(q^{(k-\delta+1)n})$ \cite{MR4149371}&$\Omega(q^{(k-\frac{\delta-1}{\alpha-1})n})$&$O(q^{\lfloor k+1-\frac{\delta}{\alpha-1}\rfloor n})$\\\hline
$B_2(n,k,k+1;3)$  &  $2^{n-2k+1}$  [Theorem \ref{lower bound for B2(n,k,k+1;3)}] & $\Omega(2^{\frac{k}{2}n})$ & $O(2^{\lfloor\frac{k+1}{2}\rfloor n})$ \\ \hline
\makecell{$B_2(n,k,k+\gamma;3)$\\$k\ge3\lfloor\frac{\gamma}{2}\rfloor$}  & $2^{\lfloor\frac{n-2k+1}{k+1}\rfloor\lfloor\frac{k}{\lfloor\gamma/2\rfloor}\rfloor}$ [Theorem \ref{lower bound for B2(n,k,k+2;3)}]& $\Omega(2^{\frac{(k+1-\gamma)n}{2}})$ & $O(2^{\lfloor\frac{k-\gamma+2}{2}\rfloor n})$  \\ \hline
\makecell{$B_2(n,k,k+\gamma;3)$\\ $\gamma+1\leq k<3\lfloor\frac{\gamma}{2}\rfloor$} & $2^{2\lfloor\frac{n-2k+1}{k+1}\rfloor}$ [Theorem \ref{lower bound for B2(n,k,2k-1;3)}] &$\Omega(2^{\frac{(k+1-\gamma)n}{2}})$ &$O(2^{\lfloor\frac{k-\gamma+2}{2}\rfloor n})$    \\ \hline
$B_2(n,k,2k;3)$ & $2^{\lfloor\frac{n-k}{k}\rfloor}$ [Theorem \ref{lower bound for B2(n,k,2k;3)}] &  $\Omega(2^{\frac
{n}{2}})$ & $O(2^{n})$   \\ \hline
\makecell{$B_q(n,k,k+\gamma;3)$\\ $q>2$} & $q^{\lfloor\frac{n-2k+1}{k+1}\rfloor\lfloor\frac{k}{\gamma}\rfloor}$ [Theorem \ref{lower bound for Bq(n,k,k+gamma;3)}] & $\Omega(q^{\frac{(k-\gamma+1)n}{2}})$ & $O(q^{\lfloor\frac{k-\gamma+2}{2}\rfloor n})$   \\ \hline
\end{tabular}
\end{center}
\end{table}


\bibliographystyle{abbrv}
\bibliography{REF}

\end{document}